\documentclass[runningheads]{llncs}
\usepackage{xspace}
\usepackage[utf8]{inputenc}
\usepackage{amssymb}
\usepackage{amsmath}
\usepackage{color}
\usepackage{cite}
\usepackage{subfig,wrapfig}
\usepackage{graphicx}
\usepackage{times}
\usepackage[bookmarks=false]{hyperref}
\usepackage[textsize=footwww.google.notesize]{todonotes}
\newcommand{\comment}[1]{{\color{red} #1}}
\newcommand{\rowsum}[1]{{\overset{\rightarrow}{#1}}}
\newcommand{\columnsum}[1]{{\downarrow\!#1}}

\newcommand{\true}{\ensuremath{\textnormal{TRUE}}}
\newcommand{\false}{\ensuremath{\textnormal{FALSE}}}

\newcommand{\bvec}[1]{\mathbf{#1}}

\newcounter{constr} 
\newcommand{\constr}[1]{\noindent%
 \refstepcounter{constr}\theconstr #1}

\title{Using ILP/SAT to determine pathwidth, visibility
representations, and other grid-based graph drawings\thanks{T. Biedl
is supported
by NSERC, M. Nöllenburg is supported by the Concept for the Future of
KIT under grant YIG 10-209.}}
\author{Therese Biedl\inst{1}
\and Thomas Bl\"asius\inst{2}
\and Benjamin Niedermann\inst{2}%
\and Martin~N\"{o}llenburg\inst{2}%
\and Roman Prutkin\inst{2}%
\and Ignaz Rutter\inst{2}%
}
\authorrunning{Biedl \and Bl\"asius \and Niedermann \and
N\"{o}llenburg \and Prutkin \and Rutter}
\titlerunning{ILP/SAT formulation for grid-based graph
representations}
\institute{David R.~Cheriton School of Computer
Science, University of Waterloo, Canada
\and
Institute of Theoretical Informatics, Karlsruhe Institute of
Technology, Germany
}

\begin{document}

\maketitle
\begin{abstract}
We present a simple and versatile formulation of grid-based graph
representation problems as an integer linear program (ILP) and a
corresponding SAT instance. In a grid-based representation vertices
and edges correspond to axis-parallel boxes on an underlying integer
grid; boxes can be further constrained in their shapes and
interactions by additional problem-specific constraints. We describe a
general $d$-dimensional model for grid representation problems. This
model can be used to solve a variety of  NP-hard graph problems,
including pathwidth, bandwidth, optimum $st$-orientation, 
area-minimal (bar-$k$) visibility representation, boxicity-$k$ graphs
and others. We implemented SAT-models for all of the above problems
and evaluated them on the Rome graphs collection. The experiments show
that our model successfully solves NP-hard problems within few minutes
on small to medium-size Rome graphs.
\end{abstract}

\section{Introduction}
Integer linear programming (ILP) and Boolean satisfiability testing
(SAT) are indispensable and widely used tools in solving many hard
combinatorial optimization and decision problems in practical
applications~\cite{bhmw-hs-09,cbd-aip-10}. In graph drawing,
especially for planar graphs, these methods are not frequently
applied. A few notable exceptions
are crossing
minimization~\cite{cmb-aecm-08,bcegjk-bacnp-08,jm-2scmpeha-97,gsm-okpcm-11}, %
orthogonal graph drawing with vertex and
edge labels~\cite{bdln-odgwvel-05} and metro-map
layout~\cite{nw-dlhqm-11}. Recent work by Chimani et
al.~\cite{cz-upt-13} uses SAT formulations for testing upward
planarity. All these approaches have in common that they exploit
problem-specific properties to derive small and efficiently solvable
models, but they do not generalize to larger classes of problems. 

In this paper we propose a generic ILP model that is flexible enough
to capture a large variety of different grid-based graph layout
problems, both polynomially-solvable and NP-complete. We demonstrate
this broad applicability by adapting the base model to six different
NP-complete example problems: pathwidth, bandwidth, optimum
$st$-orientation, minimum area bar- and bar $k$-visibility representation, and boxicity-$k$ testing. 
For minimum-area visibility representations and boxicity this is, to the best of our knowledge, the first implementation of an exact solution
method.
Of course this flexibility comes at the cost of losing some of the
efficiency of more specific approaches. Our goal, however, is not to
achieve maximal performance for a specific problem, but to provide an
easy-to-adapt solution method for a larger class of problems, which
allows quick and simple prototyping for instances that are not too
large. Our ILP models can be easily translated into equivalent SAT
formulations, which exhibit better performance in the implementation
than the ILP models themselves. We illustrate the usefulness of our
approach by an experimental evaluation that applies our generic model
to the above six NP-complete problems using the well-known Rome
graphs~\cite{rome} as a benchmark set. Our evaluation shows that,
depending on the problem, our model can solve small to medium-size
instances (sizes varying from about 25 vertices and edges for bar-$1$
visibility testing up to more than 250 vertices and edges, i.e., all
Rome graphs, for optimum $st$-orientation) to optimality within a few
minutes. 
In Section~\ref{sec:genericmodel} we introduce generic grid-based
graph representations and formulate an ILP model for $d$-dimensional
integer grids. We demonstrate how this model can be adapted to six
concrete one-, two- and $d$-dimensional grid-based graph layout
problems in Sections~\ref{se:problems:1dim} and~\ref{sec:higherdim}.
In Section~\ref{sec:exp} we evaluate our implementations and report
experimental results. The implementation is available from
\url{i11www.iti.kit.edu/gdsat}.

\section{Generic Model for Grid-Based Graph Representations}\label{sec:genericmodel}

In this section we explain how to express $d$-dimensional boxes in a
$d$-dimensional integer grid as constraints of an ILP or a SAT
instance. In the subsequent sections we use these boxes as
basic elements for representing vertices and edges in problem-specific ILP and SAT models.
We need a simple observation that shows that we can restrict ourselves
to boxes in integer grids.

\newcommand{\lemTransformToInteger}{
Any set $I$ of $n$ boxes in $\mathbb R^d$ can be transformed into another set $I'$ of~$n$ closed boxes on 
the integer grid %
$\{ 1, \ldots, n\}^d$ 
such that two boxes intersect in $I$ if and only if they intersect in $I'$.
}
\begin{lemma}
\label{lem:transformToInteger}
\label{lem:gtransformToInteger}
\lemTransformToInteger
\end{lemma}

\begin{proof}
Let $I_1,\dots,I_n$ be a set of $d$-dimensional intervals.
Then we prove that there exists a set $I'_1,\dots,I'_n$ of
$d$-dimensional intervals such that:
\begin{itemize}
\item For all $i$, $I'_i= [b_i^1,e_i^1] \times \dots \times
[b_i^d,e_i^d]$ for some
  $b_i^1,e_i^1,\dots,b_i^d,e_i^d\in \{1,\dots,n\}$.  Put differently,
  $I'_i$ has integral coordinates in the range $\{1,\dots,n\}$,
  and it is closed at both ends.
\item $I_i\cap I_j$ if and only if $I'_i\cap I'_j$.
\end{itemize}

It suffices to show this for 1-dimensional intervals; we can then
transform each dimension of the intervals separately to achieve the
results for arbitrary dimensions.
  
For 1-dimensional intervals, presume that $I_i=\{s_i,t_i\}$, where we
make no assumption over whether the ends of $I_i$ are open or closed.
We first create a sorted orientation of the $2n$ endpoints of these
intervals.  We sort them by their coordinate, and in case of a tie
take 
first right endpoints where the interval is open, 
then left endpoints where the interval is closed,
then right endpoints where the interval is closed 
and then left endpoints where the interval is open.  
Presume that $\sigma$ describes this order, i.e., for any endpoint $p$
$\sigma(p)$ gives the index of $p$ in this sorted order.
One easily verifies that $\{\sigma(s_i),\sigma(t_i)\}$ intersects 
$\{\sigma(s_j),\sigma(t_j)\}$
if and only if $I_i$ intersects $I_j$.  Here it does not even matter
whether we make the new intervals open or closed, since all endpoints
are distinct integers.

Now compact $\sigma$ by not using a new integer whenever we have
an endpoint of an interval.  More precisely, split $\sigma$ into
maximal subsequences such that each subsequence consists of multiple
(possibly none) left endpoints of intervals, followed by one right
endpoint of an interval.
  There are 
$n$ such subsequences, since there are $n$ right endpoints.
Enumerate the subsequences in order, and let $\sigma'(p)$ be
the number assigned to the subsequence that contains $\sigma(p)$,
for any endpoint $p$ of an interval.
Now for $I_i=\{s_i,t_i\}$ 
define $I'_i$ to be $\left[\sigma'(s_i),\sigma'(t_i)\right]$.  We
claim
that this satisfies the conditions.  Clearly the endpoints of the
intervals
are integers in the range $1,\dots,n$, so we only must 
argue that intersections are unchanged.
This held when going over from $\{s_i,t_i\}$ to
$\{\sigma(s_i),\sigma(t_i)\}$.
But going over from $\sigma$ to $\sigma'$, we changed the relative
order 
of endpoints only by contracting a number
of left endpoints, followed by one right endpoint.  Since 
$I'_i$ is closed at both ends, this does not change intersections.
\qed
\end{proof}

\subsection{Integer Linear Programming Model}\label{sec:ilpmodel}
We will describe our model in the general case for a $d$-dimensional
integer grid, where $d \ge 1$. Let $\mathcal R^d = [1,U_1] \times
\ldots \times [1, U_d]$ be a bounded $d$-dimensional integer grid,
where $[A,B]$ denotes the set of integers $\{A, A+1, \dots, B-1,B\}$.
In a \emph{grid-based graph representation} vertices and/or edges are
represented as $d$-dimensional boxes in $\mathcal{R}^d$. A \emph{grid
box} $R$ in $\mathcal R^d$ is a subset $[s_1,t_1] \times \ldots \times
[s_d, t_d]$ of $\mathcal R^d$, where $1 \le s_k \le t_k \le U_k$ for
all $1 \le k \le d$. In the following we describe a set of ILP
constraints that together create a non-empty box for some object  $v$.
 We denote this ILP as $\mathcal B(d)$.

We first extend $\mathcal R^d$ by a margin of dummy points to $\bar{\mathcal R}^d = [0,U_1+1] \times \ldots \times [0, U_d+1]$.
We use three sets of binary variables:
\begin{align}
  \label{eq:x} x_\bvec{i}(v) &\in \{0,1\} && \forall \bvec{i} \in \bar{\mathcal R}^d\\
  \label{eq:b} b_i^{k}(v) &\in \{0,1\} && \forall 1\le k \le d \text{ and } 1\le i \le U_k\\
  \label{eq:e} e_i^{k}(v) &\in \{0,1\} && \forall 1\le k \le d \text{ and } 1\le i \le U_k
\end{align}
The variables $x_\bvec{i}(v)$ indicate whether grid point $\bvec{i}$ belongs to the box representing~$v$ ($x_\bvec{i}(v)=1$) or not ($x_\bvec{i}(v)=0$). Variables $b_i^{k}(v)$ and $e_i^{k}(v)$ indicate whether the box of $v$ may start or end at position $i$ in dimension $k$. We use~$\bvec{i}[k]$ to denote the $k$-th coordinate of grid point $\bvec{i} \in \mathcal R^d$ and $\bvec{1_k} = (0,\dots,0,1,0,\dots,0)$ to denote the $k$-th $d$-dimensional unit vector. If $d=1$ we will drop the dimension index of the variables to simplify the notation.
The following constraints model a box in $\mathcal R^d$ (see Fig.~\ref{fig:example2d} for an example): %
\begin{align} 
\label{eq:margin} x_\bvec{i}(v) &= 0 && \forall \bvec{i} \in \bar{\mathcal R}^d \setminus \mathcal R^d\\
\label{eq:nonempty} \sum_{\bvec{i}\in \bar{\mathcal R}^d } x_\bvec{i}(v) & \ge 1 &\\
\label{eq:oneb} \sum_{{i \in [1,U_k]}} b_{i}^k(v) & = 1 &&\forall 1
\le k \le d\\
\label{eq:onee} \sum_{{i \in [1,U_k]}} e_{i}^k(v) & = 1 && \forall 1
\le k \le d
\end{align}
\begin{align}
\label{eq:start} x_\bvec{i-1_k}(v) + b_{\bvec{i}[k]}^k(v) &\ge 
x_\bvec{i}(v) && \forall \bvec{i} \in \mathcal R^d \text{ and } 1 \le
k \le d\\
\label{eq:end} x_\bvec{i}(v) &\le x_\bvec{i+1_k}(v) + e_{\bvec{i}[k]}^k(v) && \forall \bvec{i} \in \mathcal R^d \text{ and } 1 \le k \le d
\end{align}

Constraint~\eqref{eq:margin} creates a margin of zeroes
around~$\mathcal R^d$.
Constraint~\eqref{eq:nonempty} ensures that the shape representing~$v$
is non-empty, and constraints~\eqref{eq:oneb} and~\eqref{eq:onee}
provide exactly one start and end position in each dimension.
Finally, due to constraints~\eqref{eq:start} and~\eqref{eq:end} each
grid point inside the specified bounds belongs to~$v$ and all other
points don't.

\begin{figure}[tb]
	\centering
		\includegraphics[scale=1]{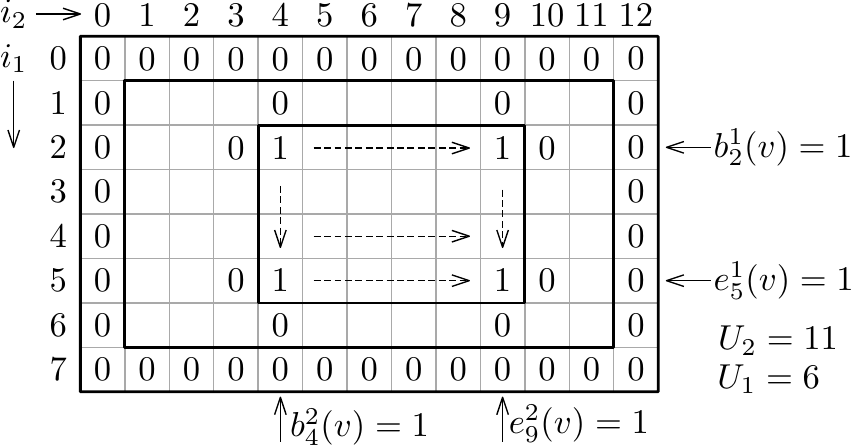}
	\caption{\small Example of a 2-dimensional $8\times 13$ grid
$\bar{\mathcal R}^d$ with a 
	$6\times 4$ grid box and the corresponding variable assignments.}
	\label{fig:example2d}
	\vspace{-2ex}
\end{figure}

\newcommand{\lemboxtext}{The ILP $\mathcal B(d)$ defined by constraints~(\ref{eq:x})--(\ref{eq:end}) correctly models all non-empty grid boxes in $\mathcal R^d$.
}

\begin{lemma}\label{lem:correctbox}
\lemboxtext
\end{lemma}

\begin{proof}
  Let $R = [s_1,t_1] \times \ldots \times [s_d, t_d]$ be a non-empty
grid box in $\mathcal R^d$. We initially set all variables to the
default value of zero. For every $\bvec{i} \in R$ we set
$x_\bvec{i}(R) = 1$. Moreover, we set $b_{s_k}^k(R) = 1$ and
$e_{t_k}^k(R) = 1$ for all $1 \le k \le d$. We claim that this
assignment satisfies all constraints. Constraints
(\ref{eq:x})--(\ref{eq:e}) are obviously satisfied, as well as 
constraints~(\ref{eq:margin}) and~(\ref{eq:nonempty}) since $R$ is
nonempty but does not intersect the margin $\bar{\mathcal R}^d
\setminus \mathcal R^d$. For every $1 \le k \le d$ we set exactly one
variable $b_i^k(R) = 1$ and $e_j^k(R)=1$, namely for $i=s_k$ and
$j=t_k$; this satisfies constraints~(\ref{eq:oneb})
and~(\ref{eq:onee}). Constraints~(\ref{eq:start}) and~(\ref{eq:end})
are trivially satisfied if $x_\bvec{i}(R) = x_\bvec{i+1_k}(R)$  for
two neighboring grid points in dimension $k$. Otherwise, if
$x_\bvec{i-1_k}(R) = 0$ and $x_\bvec{i}(R) = 1$, then
$b_{\bvec{i}[k]}^k$ must equal $1$, and if $x_\bvec{i}(R) = 1$ and
$x_\bvec{i+1_k}(R) = 0$, then $e_{\bvec{i}[k]}^k$ must equal $1$.
  Let $W_\bvec{i}^k = \{\bvec{j} \in \bar{\mathcal R^d} \mid \bvec{j}
= \bvec{i} + (\lambda - \bvec{i}[k]) \cdot \bvec{1_k}, \, 0 \le
\lambda \le U_k+1  \}$ be a row of $\mathcal R^d$ in the $k$-th
dimension indexed by a grid point~$\bvec{i}$. By $W_\bvec{i}^k[l]$ we
denote the $l$-th point $\bvec{i} + (l - \bvec{i}[k]) \cdot
\bvec{1_k}$ in row $W_\bvec{i}^k$. We know that   $x_\bvec{j}(R)=0$
for $\bvec{j} \in \{W_\bvec{i}^k[0], W_\bvec{i}^k[U_k+1]\}$. If
$W_\bvec{i}^k \cap R = \emptyset$, then all indicator variables for
$W_\bvec{i}^k$ are equal to zero and constraints~(\ref{eq:start})
and~(\ref{eq:end}) are satisfied. Otherwise, the indicator variables
for row $W_\bvec{i}^k$ contain a consecutive sequence of $1$'s for
points $W_\bvec{i}^k[l] \in R$ with $s_k \le l \le t_k$. But since
$b_{s_k}^k(R) = 1$ and $e_{t_k}^k(R) = 1$ the
constraints~(\ref{eq:start}) and~(\ref{eq:end}) are also satisfied for
$\bvec{j} = W_\bvec{i}^k[s_k]$ and $\bvec{j} = W_\bvec{i}^k[t_k]$.

Now consider a valid variable assignment according to
constraints~(\ref{eq:x})--(\ref{eq:end}) and define $R= \{\bvec{i} \in
\mathcal R^d \mid x_\bvec{i} = 1\}$. By constraints~(\ref{eq:margin})
and~(\ref{eq:nonempty}) $R$ contains at least one point. Let $1\le
k\le d$ be a dimension of $\mathcal R^d$ and let $W_\bvec{i}^k$ be any
row in the $k$-th dimension. By constraints~(\ref{eq:oneb})
and~(\ref{eq:onee}) there is exactly one coordinate $s_k$, where
$b_{s_k}^k=1$ and one coordinate $t_k$, where $e_{t_k}^k=1$. Thus by
constraints~(\ref{eq:start}) and~(\ref{eq:end}) $W_\bvec{i}^k \cap R$
is either empty or a single interval of consecutive points between
$W_\bvec{i}^k[s_k]$ and $W_\bvec{i}^k[t_k]$. Since this is true for
any $k$, the set $R$ must be a (non-empty) grid box. \qed
\end{proof}

Our example ILP models in Sections~\ref{se:problems:1dim}
and~\ref{sec:higherdim} extend ILP $\mathcal B(d)$ by additional
constraints controlling additional properties of vertex and edge
boxes. For instance, boxes can be easily constrained to be single
points, to be horizontal or vertical line segments, to intersect if
and only if they are incident or adjacent in $G$, to meet in endpoints
etc.
The definition of an objective function for the ILP depends on the
specific problem at hand and will be discussed in the problem
sections.
\subsection{Translating the ILP model into a SAT
model}\label{sec:satmodel}
In this section we explain shortly how ILP~$\mathcal B(d)$ can be
translated into an equivalent SAT formulation with better
practical performance. The transformation of $\mathcal B(d)$
(including later problem-specific extensions) into a SAT instance,
i.e., a set of Boolean clauses, is straightforward. Let $k,c \in
\mathbb N \setminus \{ 0 \}$, $k > c$, be positive integers and $y_1,
\ldots, y_k$, $z$ be binary variables. Then most of our ILP
constraints belong to one of the following four types:
(i)~$\sum_{i=1}^k y_i \geq z$, 
(ii)~\mbox{$\sum_{i=1}^k y_i \leq c$,} 
(iii)~$\sum_{i=1}^k y_i \geq c$,
(iv)~$\sum_{i=1}^k y_i = c$. 

We translate a type-(i) constraint
into the clause~$y_1 \vee \ldots \vee y_k \vee \neg z$.
For a type-(ii) constraint we consider each tuple of $c+1$ pairwise
distinct indices $i_1, \ldots, i_{c+1}$ and
 add the clause~$\neg y_{i_1} \vee \ldots \vee \neg y_{i_c} \vee \neg
y_{i_{c+1}}$. 
This gives us $\binom{k}{c+1}$ clauses of size~$c+1$.
A \mbox{type-(iii)} constraint is equivalent to~$\sum_{i=1}^k (1 -
y_i) \leq k-c$,
providing $\binom{k}{k-c+1}$ clauses of size~$k-c+1$.
A type-(iv) constraint 
is described as a type-(ii) and a type-(iii) constraint and thus needs
$\binom{k}{c+1}$ clauses of size~$c+1$
and $\binom{k}{k-c+1}$ clauses of size~$k-c+1$.

\section{One-dimensional Problems}
\label{se:problems:1dim}
In the following, let $G = (V,E)$ be an undirected graph with~$|V|=n$ and~$|E|=m$. 
One-dimensional grid-based graph representations can be used to model vertices as intersecting intervals (one-dimensional boxes) or as disjoint points that induce a certain vertex order. We present ILP models for three such problems.

\subsection{Pathwidth}

The pathwidth of a graph $G$ is a well-known graph
parameter with many equivalent definitions. We use the definition via the
smallest clique size of an interval supergraph.
More precisely, a graph is an {\em interval graph} if it can be
represented as intersection graph of 1-dimensional intervals.  A
graph $G$ has {\em pathwidth} $pw(G)\leq p$ if there exists an
interval graph $H$ that contains $G$ as a subgraph and for which all
cliques have size at most $p+1$.
It is NP-hard to compute the pathwidth of an arbitrary graph and even
hard to approximate it~\cite{BodlaenderGKH91}. There are fixed-parameter algorithms for computing the pathwidth,
e.g.~\cite{BodlaenderK96}, however, we are not aware of any
implementations of these algorithms.  The only available implementations are exponential-time algorithms, 
e.g., in
sage\footnote{www.sagemath.org}. %

\begin{problem}[Pathwidth]\label{pb:pw}
	Given a graph $G=(V,E)$, determine the pathwidth of~$G$, i.e., the smallest integer $p$ so that $pw(G) \le p$.
\end{problem}

There is an interesting connection between pathwidth and planar graph
drawings of small height.  Any planar graph that has a planar
drawing of height $h$ has pathwidth at most $h$~\cite{FLW03}.
 Also, pathwidth is a crucial ingredient in
testing in polynomial time whether a graph has a planar drawing of
height $h$~\cite{DFK+08}.

We create a one-dimensional grid representation of $G$, in which every
vertex is an interval and every edge forces the two vertex intervals
to intersect. The objective is to minimize the maximum number of
intervals that intersect in any given point.
We use the ILP $\mathcal B(1)$ for a grid $\mathcal R = [1,n]$, which
already assigns a non-empty interval to each vertex $v \in V$. We add
binary variables for the edges of $G$, a variable $p \in \mathbb N$
representing the pathwidth of $G$, and a set of additional
constraints as follows.
\begin{align}
  \label{eq:pw:edgevar} x_{i}(e) &\in \{0,1\} &&&& \forall {i} \in {\mathcal R} \text{ and } e \in E\\
  \label{eq:pw:oneedge} \sum_{i \in \mathcal R} x_{i}(e) &\ge 1 &&&& \forall e \in E\\
  \label{eq:pw:edge} x_i(uv) &\le x_i(u), & x_i(uv) &\le x_i(v) && \forall uv \in E\\
  \label{eq:pw:pw} \sum_{v \in V} x_i(v) &\le p+1 &&&& \forall i \in \mathcal R 
\end{align}
Our objective function is to minimize the value of $p$ subject to the above constraints.

It is easy to see that every edge must be represented by some grid
point (constraint~(\ref{eq:pw:oneedge})), and can only use those grid
points, where the two end vertices intersect
(constraint~(\ref{eq:pw:edge})). Hence the intervals of vertices
define some interval graph $H$ that is a supergraph of $G$. 
Constraint (13) enforces that at most $p+1$ intervals meet in any
point, which by Helly's property means that $H$ has clique-size at
most $p+1$. So $G$ has pathwidth at most $p$. 
By minimizing $p$ we obtain the desired result. 
In our implementation we translate the ILP into a SAT instance
using the rules given in Section~\ref{sec:satmodel}. 
We test satisfiability
for fixed values of $p$, starting with $p=1$ and increasing it
incrementally until a solution is found.
\begin{theorem}
There exists an ILP/SAT formulation with $O(n(n+m))$ variables
and $O(n(n+m))$ constraints / $O(n^3 + n \binom{n}{p+2})$ clauses of
maximum size $n$ that has a solution of value $\leq p$ if and only
if $G$ has pathwidth $\leq p$.
\end{theorem}

\par
With some easy modifications, the above ILP can be used 
for testing whether a graph is a (proper) interval graph. Section~\ref{sec:boxicity} shows that boxicity-$d$ graphs, the $d$-dimensional generalization of interval graphs, can also be recognized by our ILP.

\subsection{Bandwidth}\label{sec:bandwidth}
The bandwidth of a graph $G$ with $n$ vertices is another classic graph parameter, which is NP-hard to compute~\cite{ccdg-bpgms-82}; due to the practical importance of the problem there are also a few approaches to find exact solutions to the bandwidth minimization problem.
For example,~\cite{bandwidth-exact} and~\cite{mcp-bbambm-08} use the
branch-and-bound technique combined with various heuristics. We
present a solution that can be easily described using our general
framework. However, regarding the running time, it cannot be expected
to compete with techniques specially tuned for solving the bandwidth
minimization problem.

Let $f: V \rightarrow \{1, \dots, n\}$ be a bijection that defines a
linear vertex order. The \emph{bandwidth} of $G$ is defined as $bw(G)
= \min_f \max \{f(v) - f(u) \mid uv \in E \text{ and } f(u) < f(v)\}$,
i.e., the minimum length of the longest edge in $G$ over all possible
vertex orders.

We describe an ILP that assigns the vertices of $G$ to disjoint grid
points  and requires for an integer $k$ that any pair of adjacent
vertices is at most $k$ grid points apart. If the ILP has a solution,
then we know $bw(G) \le k$. We set $\mathcal R = [1, n]$ and add the
following two constraints:
\begin{align}
  \label{eq:bw:onlyone} \sum_{v \in V} x_i(v) &\le 1 && \forall i \in
\mathcal R\\
  \label{eq:bw:dist} x_i(u) &\le \sum_{j=i-k}^{i+k} x_j(v) && \forall
uv \in E\,  \forall i \in \mathcal R
\end{align}
Constraint~(\ref{eq:bw:onlyone}) guarantees that no grid point is
occupied by more than one vertex, and constraint~(\ref{eq:bw:dist})
requires that any two adjacent vertices are at most $k$ grid points
apart.
We note that the variables $b_i(v)$ and $e_i(v)$ and their constraints
are not required in this model since in $\mathcal R = [1,n]$
constraints~(\ref{eq:nonempty}) and~(\ref{eq:bw:onlyone}) suffice to
set exactly one variable $x_i(v)=1$ for every $v \in V$.
We do not need an objective function but rather test if the feasible
region is non-empty for a given $k\ge 1$.

\begin{theorem}
There exists an ILP/SAT formulation with $O(n^2)$ variables
and $O(n \cdot m)$ constraints / $O(n^3)$ clauses of maximum size $n$
that has a solution if and only if
$G$ has bandwidth $\leq k$. 
\end{theorem}

\subsection{Optimum $st$-orientation}

Let $G$ be an undirected graph and let $s$ and $t$ be two vertices of
$G$ with $st \in E$.  An {\em $st$-orientation} of $G$ is an
orientation of the edges such
that $s$ is the unique source and $t$ is the unique sink~\cite{ET76}.  Such an
orientation can exist only if $G$ is biconnected. Computing an
$st$-orientation can be done in linear time~\cite{ET76,b-eo-02}, but
it is NP-complete to find an $st$-orientation that minimizes the
length of the longest path from $s$ to $t$, even for planar
graphs~\cite{SZ10}. It has many
applications in graph drawing~\cite{pt-apsto-10} and beyond.

\begin{problem}[Optimum $st$-orientation]\label{pb:st}
	Given a graph $G=(V,E)$ and two vertices $s,t \in V$ with $st \in E$, find an orientation of $E$ such that $s$ is the only source, $t$ is the only sink, and the length of the longest directed path from $s$ to $t$ is minimum.
\end{problem}

We now formulate an ILP that computes a \emph{height-$k$
$st$-orientation} of $G$, i.e., an $st$-orientation such that the
longest path has length at most $k$ (if one exists).
We use the ILP $\mathcal B(1)$ for $\mathcal R = [1,k]$ to assign
intervals to the vertices and edges of $G$. Vertices are required to
occupy exactly one point, whereas edges must span at least two points.
The additional constraints are as follows:
\begin{align}
  \label{eq:st:one} \sum_{i \in \mathcal R} x_{i}(v) &\le 1 &&&&
\forall v \in V\\
  \label{eq:st:edge} \sum_{i \in \mathcal R} x_i(e) &\ge 2 &&&&
\forall e \in E\\
  \label{eq:st:beginend} b_i(uv) &\le x_i(u) + x_i(v), & e_i(uv)&\le
x_i(u) + x_i(v) && \forall i \in \mathcal R\, \forall uv \in E \\
  \label{eq:st:out} x_i(v) &\le \sum_{vw \in E} b_i(vw) &&&& \forall i
\in \mathcal R \, \forall v \in V\setminus\{t\}\\
  \label{eq:st:in} x_i(v) &\le \sum_{vw \in E} e_i(vw) &&&& \forall i
\in \mathcal R \, \forall v \in V\setminus\{s\}
\end{align}
Similarly to the ILP in Section~\ref{sec:bandwidth}, the variables
$b_i(v)$ and $e_i(v)$ and their constraints are not required for the
vertices, since constraints~(\ref{eq:nonempty})
and~(\ref{eq:st:one}) ensure that every vertex interval consists of a
single grid point. Constraint~(\ref{eq:st:edge}) guarantees that each
edge interval contains at least two points.
Constraint~(\ref{eq:st:beginend}) makes sure that every edge must
begin and end at the grid points occupied by its end vertices.
Finally, constraints~(\ref{eq:st:out}) and~(\ref{eq:st:in}) ensure
that every vertex except $t$ has an outgoing edge and every vertex
except $s$ has an incoming edge. Since $\mathcal R = [1,k]$ and
adjacent vertices cannot occupy the same grid point, we know that if
this ILP has a feasible solution, then the longest directed path from
$s$ to $t$ has length at most $k$. Alternatively, we can set $k=n$ and
use the objective function $\min \sum_{i=1}^n i\cdot x_i(t)$ to
minimize the position of the sink $t$ and thus the longest path from
$s$ to $t$.

\begin{theorem}
There exists an ILP with $O(n(n+m))$ variables
and constraints that computes an optimum $st$-orientation. 
Alternatively, there exists an ILP/SAT formulation with $O(k(n+m))$
variables and $O(k(n+m))$ constraints / $O(k^2(n+m))$ clauses of
maximum size~$n$ that has a solution if and only if a height-$k$
$st$-orientation of~$G$ exists. 
\end{theorem}

\section{Higher-Dimensional Problems}\label{sec:higherdim}
In this section we give examples of two-dimensional visibility graph representations and a $d$-dimensional grid-based graph representation problem. Let again $G=(V,E)$ be an undirected graph with $|V|=n$ and $|E|=m$.

\subsection{Visibility representations}

A visibility representation (also: \emph{bar visibility representation} or \emph{weak visibility representation}) of a graph $G=(V,E)$ maps all vertices to disjoint horizontal line segments, called \emph{bars}, and all edges to disjoint vertical bars, such that for each edge $uv \in E$ the bar of $uv$ has its endpoints on the bars for $u$ and $v$ and does not intersect any other vertex bar. Visibility representations are an important visualization concept in graph drawing, e.g., it is well known that a graph is planar if and only if it has a visibility representation~\cite{Wis85,TT86}.
An interesting recent extension are bar $k$-visibility representations~\cite{deglst-bkvg-07}, which additionally allow edges to intersect at most $k$ non-incident vertex bars. We use our ILP to compute compact visibility and bar $k$-visibility representations. Minimizing the area of a visibility representation is NP-hard~\cite{le-tardhpg-03} and we are not aware of any implemented exact algorithms to solve the problem for any $k\ge 0$. 
By Lemma~\ref{lem:transformToInteger} we
know that all bars can be described with integer coordinates of size $O(m+n)$.

\begin{problem}[Bar $k$-Visibility Representation]\label{pb:barvis}
	Given a graph $G$, an integer grid of size $H\times W$, and an integer $k\ge 0$, find a bar $k$-visibility representation on the given grid (if one exists).
\end{problem}

\paragraph{Bar visibility representations.}
Our goal is to test whether $G$ has a visibility representation in a grid with $H$ columns and $W$ rows (and thus minimize $H$ or $W$). We set $\mathcal R^2 = [1,H]\times[1,W]$ and use ILP $\mathcal B(2)$ to create grid boxes for all edges and vertices in $G$. We add one more set of binary variables for vertex-edge incidences and the following constraints.
\begin{align}
	\label{eq:vi:inci} x_\bvec{i}(e,v) &\in \{0,1\} &&&& \forall \bvec{i} \in \mathcal R^2\, \forall e \in E\, \forall v \in e\\
  \label{eq:vi:vhor} b_i^1(v) &= e_i^1(v) &&&& \forall i \in [1,U_1]\, \forall v \in V\\
  \label{eq:vi:ever} b_i^2(e) &= e_i^2(e) &&&& \forall i \in [1,U_2]\, \forall e \in E\\
  \label{eq:vi:disj} \sum_{v \in V} x_\bvec{i}(v) &\le 1 &&&& \forall \bvec{i} \in \mathcal R^2\\
\label{eq:vi:edge} \sum_{v \in V\setminus e} x_\bvec{i}(v) &\le (1-x_\bvec{i}(e)) &&&& \forall \bvec{i} \in \mathcal R^2\, \forall e \in E
\end{align}
\begin{align}
  \label{eq:vi:inci2} x_\bvec{i}(e,v) &\le x_\bvec{i}(e), & x_\bvec{i}(e,v) &\le x_\bvec{i}(v) && \forall \bvec{i} \in \mathcal R^2\, \forall e \in E\, \forall v \in e\\
  \label{eq:vi:inci3} \sum_{\bvec{i} \in \mathcal R^2} x_\bvec{i}(e,v) &\ge 1 &&&& \forall e \in E\, \forall v \in e\\
  \label{eq:vi:sten} x_\bvec{i}(e,v) &\le b_{\bvec{i}[1]}^1(e) + e_{\bvec{i}[1]}^1(e) &&&& \forall \bvec{i} \in \mathcal R^2\, \forall e \in E\, \forall v \in e
\end{align}
Constraints~(\ref{eq:vi:vhor}) and~(\ref{eq:vi:ever}) ensure that 
all vertex boxes are horizontal bars of height 1 
and all edge boxes are vertical bars of width 1. 
Constraint~(\ref{eq:vi:disj}) forces the vertex boxes to be disjoint; 
edge boxes will be implicitly disjoint due to the remaining constraints. 
No edge is allowed to intersect a non-incident vertex 
due to constraint~(\ref{eq:vi:edge}). 
Finally, we need to set the new incidence variables $x_\bvec{i}(e,v)$ 
for an edge $e$ and an incident vertex $v$ so that $x_\bvec{i}(e,v) = 1$ 
if and only if $e$ and $v$ share the grid point $\bvec{i}$. 
Constraints~(\ref{eq:vi:inci2}) and~(\ref{eq:vi:inci3}) ensure that 
each incidence in $G$ is realized in at least one grid point, 
but it must be one that is used by the boxes of $e$ and $v$. 
Finally, constraint~(\ref{eq:vi:sten}) requires edge $e$ 
to start and end at its two intersection points 
with the incident vertex boxes. 
This constraint is optional, 
but yields a tighter formulation.%
To transform constraint~(\ref{eq:vi:edge}) into SAT,
we add the clause $\neg x_\bvec{i}(e) \vee \neg x_\bvec{i}(v)$
for each $\bvec{i} \in \mathcal R^2$, $e \in E$, $v \in V \setminus e$, and 
apply the general transformation rules otherwise.

Since every graph with a visibility representation is planar (and vice versa) we have $m \in
O(n)$. Moreover, our ILP and SAT models can also be used to test planarity of a
given graph by setting $H= n$ and $W= 2n-4$, which is sufficient due to Tamassia and Tollis~\cite{TT86}. This might not look interesting at first sight since planarity testing can be done in linear time~\cite{planarityTest}. However, we think that this is still useful as one can add other constraints to the ILP model, e.g., 
to create simultaneous planar embeddings, and use it as a subroutine for ILP formulations of applied graph drawing problems
such as metro maps~\cite{nw-dlhqm-11} and cartograms.

\begin{theorem}\label{thm:barvis}
	There exists an ILP/SAT formulation with $O(HWn)$ variables and
$O(HWn)$ constraints / $O (H W n^2)$ clauses of maximum size $HW$ 
that solves Problem~\ref{pb:barvis} for $k=0$. 
\end{theorem}

\paragraph{Bar $k$-visibility representations.}
It is easy to extend our previous model for $k=0$ to test bar
$k$-visibility representations for $k\ge 1$. We drop
constraint~(\ref{eq:vi:edge}), introduce another set of binary
variables to indicate intersections between edges and non-incident
vertices, and add the following constraints.
\begin{align}
  \label{eq:kvi:vars} y_\bvec{i}(e) &\in \{0,1\} && \forall \bvec{i}
\in \mathcal R^2\, \forall e \in E\\
  \label{eq:kvi:cross} x_\bvec{i}(e) + \sum_{v \in V\setminus e}
x_\bvec{i}(v) &\le y_\bvec{i}(e)+1 && \forall \bvec{i} \in \mathcal
R^2\, \forall e \in E\\
  \label{eq:kvi:bark} \sum_{\bvec{i} \in \mathcal R^2} y_\bvec{i}(e)
&\le k && \forall e \in E
\end{align}
\vspace{-2ex}
\begin{equation}
    \label{eq:kvi:disj} x_\bvec{i}(e) + x_\bvec{i}(e') - \frac{1}{2}
(b_{\bvec{i}[1]}^1(e) + e_{\bvec{i}[1]}^1(e) + b_{\bvec{i}[1]}^1(e') +
e_{\bvec{i}[1]}^1(e')) \le 1 \, \forall \bvec{i} \in \mathcal R^2\,
\forall (e,e') \in \binom{E}{2}
\end{equation}
The variable $y_\bvec{i}(e)$ is supposed to equal 1 if and only if $e$
intersects a non-incident vertex at position $\bvec{i}$.
Constraint~(\ref{eq:kvi:cross}) enforces that $y_\bvec{i}(e)=1$ if
grid point $\bvec{i}$ is occupied by $e$ and a non-incident vertex and
constraint~(\ref{eq:kvi:bark}) makes sure that no more than $k$ such
non-incident bars are crossed by each edge. Finally,
constraint~(\ref{eq:kvi:disj}) guarantees that any two edge bars are
disjoint, except for the case that they are incident to the same
vertex bar and meet at a common endpoint. (Alternatively, we could
enforce disjointness for all edge bars if we required vertex bars of
height~$2$.)

To transform constraint~(\ref{eq:kvi:cross}) into SAT,
we add clause $\neg x_\bvec{i}(e) \vee \neg x_\bvec{i}(v) \vee
y_\bvec{i}(e)$
for each $\bvec{i} \in \mathcal R^2$, $e \in E$, $v \notin e$.
For~(\ref{eq:kvi:disj}), we add
$\neg x_\bvec{i}(e) \vee \neg x_\bvec{i}(e') \vee b_{\bvec{i}[1]}^1(e)
\vee e_{\bvec{i}[1]}^1(e)$
for each ordered pair $e,e' \in E$, $e \neq e'$.
For the rest of the constraints, we apply the general transformation
rules.

\begin{theorem}\label{thm:barkvis}
	There exists an ILP/SAT formulation with $O(HW(n+m))$ variables and
$O(HW(m^2+n))$ constraints / $O(\binom{HW}{k+1} m + HW m^2)$ clauses
of maximum size $HW$ that solves Problem~\ref{pb:barvis} for $k\ge 1$.
\end{theorem}

\subsection{Boxicity-$d$ graphs}\label{sec:boxicity}
A graph is said to have {\em boxicity} $d$ if it can be represented
as intersection graph of $d$-dimensional axis-aligned boxes. Testing
whether a graph has boxicity $d$ is NP-hard, even for
$d=2$~\cite{Kra94}. We are not aware of any implemented algorithms to
determine the boxicity of a graph. By
Lemma~\ref{lem:transformToInteger} we can restrict ourselves to a grid
of side length $n$.

We use ILP $\mathcal B(d)$ for $\mathcal R^d = [1,n]^d$ to model a
$d$-dimensional box for each vertex of a graph $G=(V,E)$. We add the
following variables and constraints to achieve the correct
intersection properties.
\begin{align}
  \label{eq:bx:var} x_\bvec{i}(e) &\in \{0,1\} &&&& \forall \bvec{i}
\in \mathcal R^d\, \forall e \in E\\
  \label{eq:bx:nonempty} \sum_{\bvec{i} \in \mathcal R^d}
x_\bvec{i}(e) &\ge 1 &&&& \forall e \in E\\
  \label{eq:bx:vtcs} x_\bvec{i}(uv) &\le x_\bvec{i}(u), &
x_\bvec{i}(uv) &\le x_\bvec{i}(v) && \forall \bvec{i} \in \mathcal
R^d\, \forall uv \in E\\
  \label{eq:bx:disj} x_\bvec{i}(u) + x_\bvec{i}(v) &\le 1 &&&& \forall
\bvec{i} \in \mathcal R^d\, \forall uv \in \binom{V}{2} \setminus E
\end{align}
The variables $x_\bvec{i}(e)$ indicate whether a grid point lies in
the intersection of two boxes and thus represents an edge.
Constraint~(\ref{eq:bx:nonempty}) guarantees that there is a non-empty
intersection for each edge $e \in E$, and by
constraint~(\ref{eq:bx:vtcs}) we make sure that the intersection
indicator for an edge $uv$ can only be set to $1$ at position
$\bvec{i}$ if the grid boxes for $u$ and $v$ both occupy point
$\bvec{i}$. Finally, constraint~(\ref{eq:bx:disj}) enforces that no
grid point can be occupied by a pair of non-adjacent vertices.

\begin{theorem}
There exists an ILP with $O(n^d (n+m))$ variables and
$O(n^{d+2})$ constraints 
as well as a SAT instance with $O(n^d (n+m))$ variables and
$O(n^{d+2})$ clauses of maximum size $O(n^d)$
to test whether a graph $G$ has boxicity $d$.
\end{theorem}

\section{Experiments}\label{sec:exp}
We implemented and tested our formulation for minimizing pathwidth, bandwidth,
length of longest path in an $st$-orientation, and width of
bar-visibility and bar 1-visibility representations, as well as
deciding whether a graph has boxicity~2.
\par

The experiments were performed on a single core of an
AMD Opteron 6172 processor running Linux~3.4.11. The machine is
clocked at 2.1 Ghz, and has 256 GiB RAM. Our
implementation\footnote{available from
\url{http://i11www.iti.kit.edu/gdsat}} is written
in C++ and was compiled with GCC 4.7.1 using optimization
\texttt{-O3}. As test sample we used the \emph{Rome graphs} dataset
\cite{rome} which consists of 11533 graphs with vertex number between
10 and 100. $18\%$ of the Rome graphs are planar.
Figure~\ref{fig:size-distr} shows
the size distribution of the Rome graphs.

\begin{figure}[tbp]
\centering
  \includegraphics{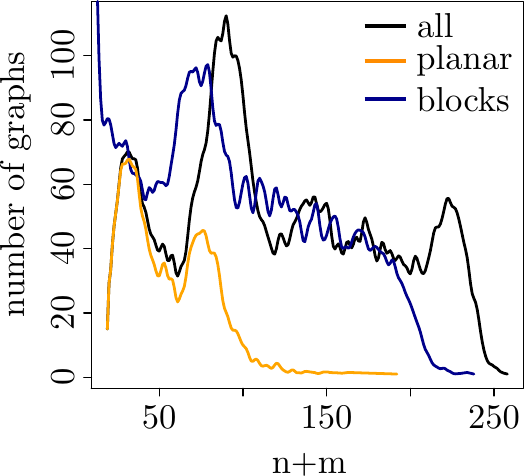}
 \caption{\small Size distribution of Rome graphs.
 Black: all graphs,
 blue: biconnected blocks ($n{\geq}3$),
 orange: planar graphs. We used an $(1,2,3,2,1)$-Gaussian filter to
reduce oscillations.  
 }
 \label{fig:size-distr}
\end{figure}

We initially used the \emph{Gurobi} solver \cite{gurobi} to
test the implementation of the ILP formulations, however it turned out
that even for very small graphs ($n<10$) solving a single instance can
take minutes. We therefore focused on the equivalent SAT
formulations gaining a significant speed-up. As SAT solver we used
\emph{MiniSat}~\cite{minisat} in version~2.2.0.  
For each of the five minimization problems we determined obvious lower and upper bounds in~$O(n)$ for the respective graph parameter. Starting with the lower bound we iteratively increased the parameter to the next integer until a solution was found (or a predefined timeout was exceeded). Each iteration consists of constructing the SAT formulation 
and executing the SAT solver. We measured the total time spent in all iterations.
For boxicity~2 we decided to consider square grids and minimize
their side lengths. Thus the same iterative procedure applies to boxicity~2.

Note that for all considered
problems a binary search-like procedure for the parameter value did not prove to be efficient, since the solver usually takes more time with increasing parameter value, which is mainly due to the increasing number of variables and
clauses.
For the one-dimensional problems we used a timeout of 300 seconds, for the two-dimensional problems of 600 seconds. 

We ran the instances sorted by size~$n+m$ starting with the smallest
graphs. If more than 400 consecutive graphs in this order produced
timeouts, we ended the experiment prematurely and evaluated only the
so far obtained results. Figures~\ref{fig:1d} and~\ref{fig:2d}
summarize our experimental results and show the percentage of Rome
graphs solved within the given time limit, as well as scatter plots
with each solved instance represented as a point depending on its
graph size and the required computation time. 

\begin{figure}[t]
 \subfloat[percentage of solved instances]{
 	\includegraphics{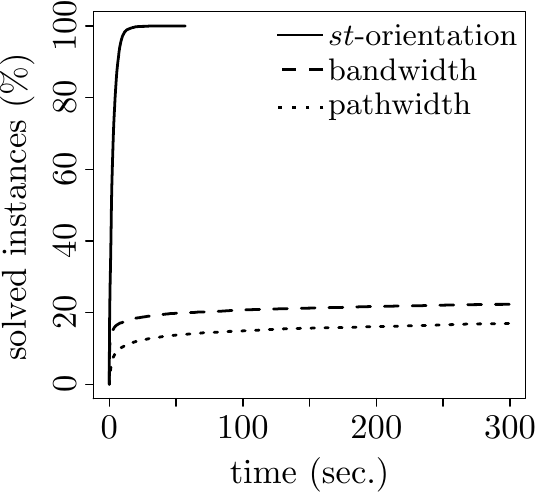}
 	\label{fig:solved}
 }
 \hfill
 \subfloat[pathwidth]{
 \includegraphics
 {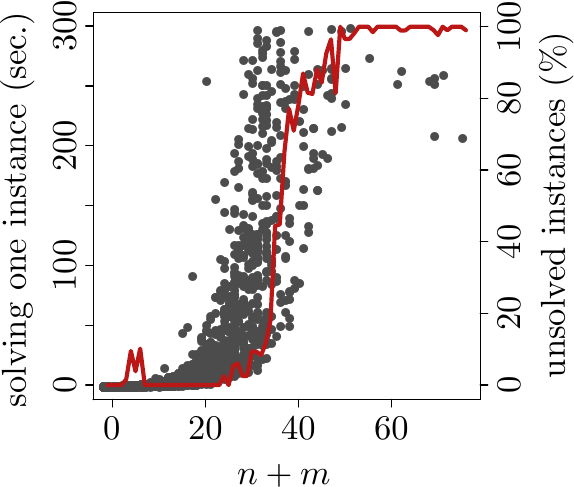}
 	\label{fig:pw:time}
 }\\
  \subfloat[bandwidth]{
 	\includegraphics{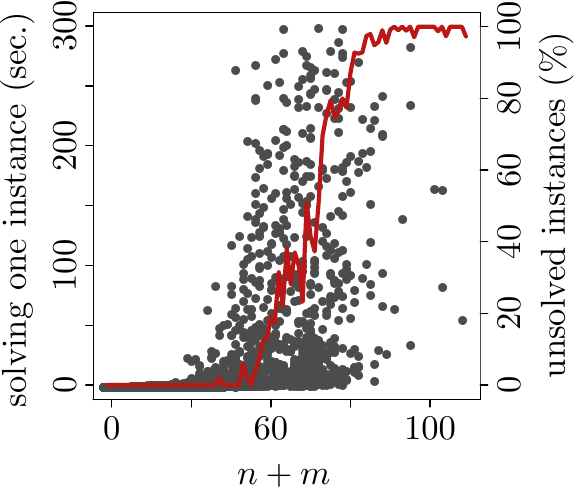}
 	\label{fig:bw:time}
 }
 \hfill
 \subfloat[$st$-orientation]{
 	\includegraphics{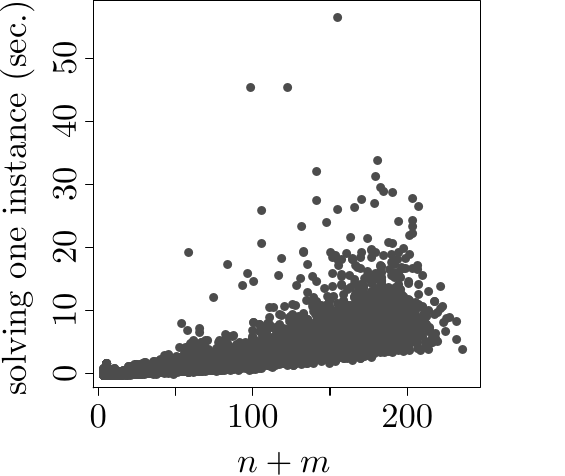}
 	\label{fig:st:time}
 }
 \caption{\small
  Experimental results for the one-dimensional problems.
  \protect\subref{fig:solved} Percentage of solved instances.
  \protect\subref{fig:pw:time}--\protect\subref{fig:st:time}: 
  Time in seconds for solving an instance (dots)
  and percentage of instances not solved within 300 seconds (red curves), 
  both in relation to~$n+m$ 
 }\label{fig:1d}
	\vspace{-2ex}
\end{figure}

\smallskip\noindent\textit{Pathwidth.}
As Fig.~\ref{fig:solved} shows,
we were able to compute the pathwidth for $17.0\%$ of all Rome graphs,
from which $82\%$ were solved within the first minute
and only $3\%$ within the last.
Therefore, we expect that a significant increase of the timeout value
would be necessary for a noticeable increase 
of the percentage of solved instances.
We note that almost all small graphs ($n+m<45$)
could be solved within the given timeout,
however, for larger graphs, the percentage of solved instances rapidly drops,
as the red curve in Fig.~\ref{fig:pw:time} shows.
Almost no graphs with $n+m>70$ were solved.

\smallskip\noindent\textit{Bandwidth.}
We were able to compute the bandwidth for $22.3\%$ of all Rome graphs
(see Fig.~\ref{fig:solved}),
from which $90\%$ were solved within the first minute
and only $1.3\%$ within the last.
Similarly to the previous case, the procedure 
terminated successfully within 300 seconds 
for almost all small graphs ($n+m<55$ in this case), 
while almost none of the larger graphs ($n+m>80$) were solved;
see the red curve in Fig.~\ref{fig:bw:time}.

\smallskip\noindent\textit{Optimum $st$-orientation.}
Note that very few of the Rome graphs are biconnected.
Therefore, to test our SAT implementation 
for computing the minimum number of levels in an $st$-orientation,
we subdivided each graph into biconnected blocks 
and removed those with~$n \leq 2$,
which produced 13606 blocks in total;
see Fig.~\ref{fig:size-distr} for the distribution of block sizes.
Then, for each such block, we randomly selected 
one pair of vertices~$s,t$, $s \neq t$, connected them by an edge if
it did not already exist and ran the iterative procedure.
In this way, for the respective choice of~$s,t$
we were able to compute the minimum number of levels 
in an $st$-orientation for all biconnected blocks; 
see Fig.~\ref{fig:solved}.
Moreover, no graph took longer than 57 seconds,
for $97\%$ of the graphs it took less than 10 seconds
and for $68\%$ less than 3 seconds.
Even for the biggest blocks with $m+n > 200$,
the procedure successfully terminated within 15 seconds 
in $93\%$ of the cases; see Fig.~\ref{fig:st:time}.

\begin{figure}[t]
 \subfloat[percentage of solved instances]{
 	\includegraphics{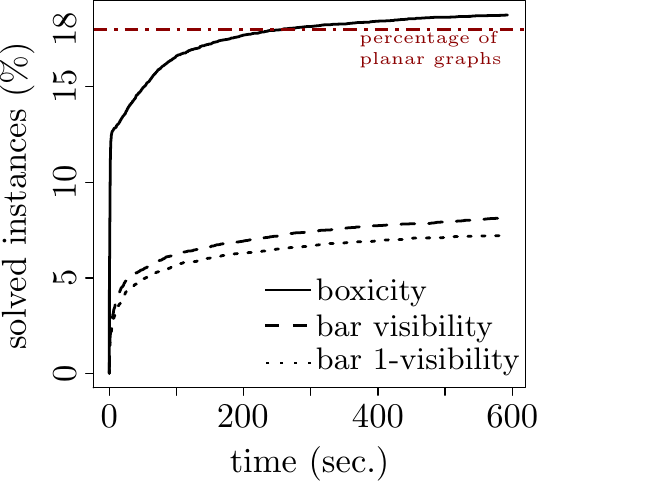}
 	\label{fig:visrep:solved}
 }
 \hfill
 \subfloat[bar visibility]{
\includegraphics{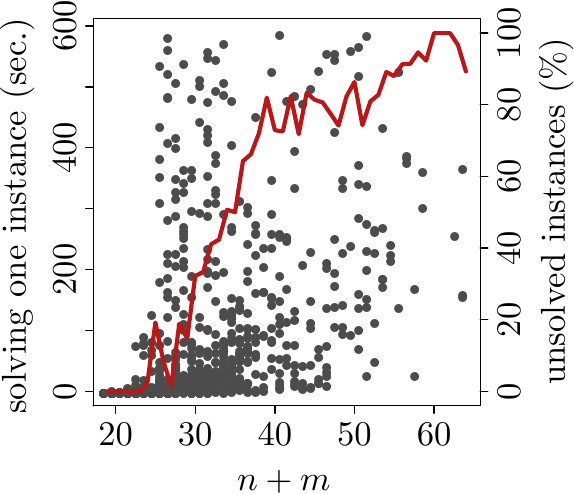}
 	\label{fig:visrep:time}
 }\\
 \subfloat[bar 1-visibility]{
 	\includegraphics{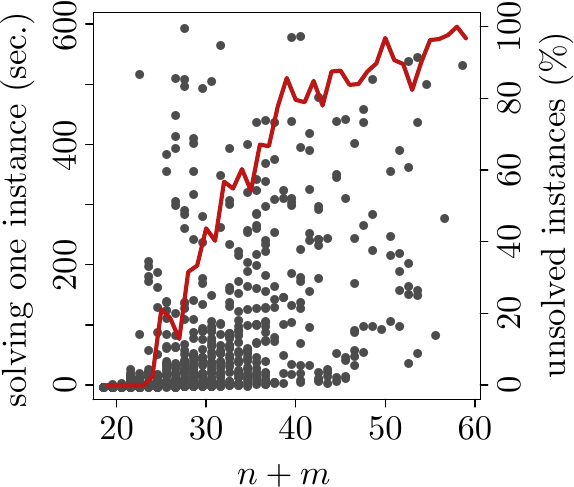}
 	\label{fig:1bar:time}
 }
 \hfill
  \subfloat[boxicity 2]{
 	\includegraphics{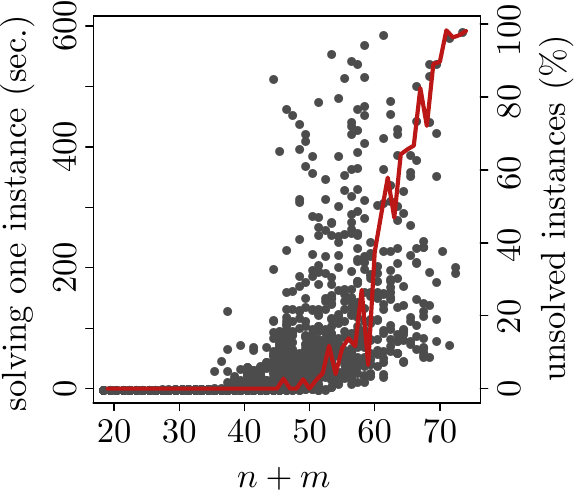}
 	\label{fig:box:time}
 }
 \caption{\small Experimental results for the two-dimensional problems.
  \protect\subref{fig:visrep:solved} 
   Percentage of solved instances. The red horizontal line shows the
   percentage of planar graphs over all Rome graphs.
  \protect\subref{fig:visrep:time}--\protect\subref{fig:box:time}:  
   Time in seconds for solving an instance (dots)
   and percentage of instances not solved within 600 seconds (red curves), 
   both in relation to~$n+m$.
 }\label{fig:2d}
	\vspace{-2ex}
\end{figure}

\smallskip\noindent\textit{Bar visibility.}
To compute bar-visibility representations of minimum width,
we iteratively tested for each graph all widths~$W$
between~1 and~n. We used the trivial upper bound $H=n$ for the height.
We were able to compute solutions
for~$28.5\%$ of all 3281 planar Rome graphs (see
Fig.~\ref{fig:visrep:solved}), $69\%$ of which were solved within the
first minute and less than~$0.1\%$ within the last.
We were able to solve all small instances with~$n+m \leq 23$
and almost none for~$n+m >55$; see the red curve in
Fig.~\ref{fig:visrep:time}.
\par

\smallskip\noindent\textit{Bar 1-visibility.} %
We also ran the width minimization procedure for 
bar 1-visibility representations on all Rome graphs. 
The number of graphs for which the procedure terminated successfully
within the given time is 833 ($7.2\%$ of all Rome graphs),
which is close to the corresponding number for bar-visibility; see
Fig.~\ref{fig:visrep:solved}. 
 This is not surprising, since most small Rome graphs are planar;
 see Fig.~\ref{fig:size-distr}. 
For bar 1-visibility, eight graphs were solved 
which were not solved for bar-visibility.
Interestingly, they were all planar.
All but 113 graphs successfully processed in the previous
experiment 
were also successfully processed in this one. 
A possible explanation for those~113 graphs is that the SAT formulation for
 bar 1-visibility requires more clauses.
All small graphs with $n+m \leq 23$ were processed successfully.
Interestingly, for none of the processed graphs
the minimum width actually decreased in comparison to their
minimum-width bar-visibility representation.

\smallskip\noindent\textit{Boxicity-2}
For testing boxicity~2,
we started with a~$3\times 3$ grid for each graph 
and then increased height and width simultaneously after each iteration.
Within the specified timeout of~600 seconds,
we were able to decide whether a graph has boxicity~2
for $18.7\%$ of all Rome graphs (see Fig.~\ref{fig:visrep:time}),
$82\%$ of which were processed within the first minute and~$0.3\%$
within the last.
All of the successfully processed graphs actually had boxicity~2.
Small graphs with $n+m \leq 50$ were processed almost completely,
while almost none of the graphs with $n+m> 70$ finished;
see Fig.~\ref{fig:box:time}. 
\section{Conclusion}
We presented a versatile ILP formulation for determining placement of
grid boxes according to problem-specific constraints. We gave six
examples of how to extend this formulation for solving numerous
NP-hard graph drawing and graph representation problems, such as
bar-visibility representations, computing the pathwidth and the
boxicity, and finding an st-orientation that minimizes the
longest directed path. Our experimental evaluation showed that while solving the original ILP is rather slow, the easily derived SAT formulations perform quite well for smaller graphs. While our approach is not suitable to replace specialized exact or heuristic algorithms that are faster and/or can solve larger instances of these problems, it does provide a simple-to-use tool for solving problems that can be modeled by grid-based graph representations without much implementation effort. This can be useful, e.g., for verifying counterexamples, NP-hardness gadgets, or for computing solutions to certain instances in practice.

We note that many other problems can easily be formulated as ILPs by assigning
grid-boxes to vertices or edges.  Among those are, e.g., testing
whether a planar graph has a straight-line drawing of height $h$,
testing whether a planar graph has a rectangular dual with
integer coordinates and prescribed integral areas, testing whether a graph is a $t$-interval graph, or whether a  bipartite graph can be represented as a planar bus graph. 
Important open problems are to reduce the complexity of our formulations and the question whether approximation algorithms for graph drawing can be derived from our model via fractional relaxation. 

{\small
  \bibliography{abbrv,ilp-gd}

\begin{thebibliography}{10}

\bibitem{rome}
Rome graphs.
\newblock \url{www.graphdrawing.org/download/rome-graphml.tgz}.

\bibitem{bhmw-hs-09}
A.~Biere, M.~Heule, H.~van Maaren, and T.~Walsh, editors.
\newblock {\em Handbook of Satisfiability}.
\newblock IOS Press, 2009.

\bibitem{bdln-odgwvel-05}
C.~Binucci, W.~Didimo, G.~Liotta, and M.~Nonato.
\newblock Orthogonal drawings of graphs with vertex and edge labels.
\newblock {\em Comput. Geom. Theory Appl.}, 32(2):71--114, 2005.

\bibitem{BodlaenderGKH91}
H.~L. Bodlaender, J.~R. Gilbert, T.~Kloks, and H.~Hafsteinsson.
\newblock Approximating treewidth, pathwidth, and minimum elimination tree
  height.
\newblock In {\em Proc. Workshop Graph-Theoretic Concepts Comput. Sci.
  (WG'91)}, volume 570 of {\em LNCS}, pages 1--12. Springer, 1992.

\bibitem{BodlaenderK96}
H.~L. Bodlaender and T.~Kloks.
\newblock Efficient and constructive algorithms for the pathwidth and treewidth
  of graphs.
\newblock {\em J. Algorithms}, 21(2):358--402, 1996.

\bibitem{b-eo-02}
U.~Brandes.
\newblock Eager $st$ ordering.
\newblock In {\em Europ. Symp. Algorithms (ESA'02)}, volume 2461 of {\em LNCS},
  pages 247--256. Springer, 2002.

\bibitem{bcegjk-bacnp-08}
C.~Buchheim, M.~Chimani, D.~Ebner, C.~Gutwenger, M.~J{\"u}nger, G.~W. Klau,
  P.~Mutzel, and R.~Weiskircher.
\newblock A branch-and-cut approach to the crossing number problem.
\newblock {\em Discrete Optimization}, 5(2):373--388, 2008.

\bibitem{cbd-aip-10}
D.-S. Chen, R.~G. Batson, and Y.~Dang.
\newblock {\em Applied Integer Programming}.
\newblock Wiley, 2010.

\bibitem{cmb-aecm-08}
M.~Chimani, P.~Mutzel, and I.~Bomze.
\newblock A new approach to exact crossing minimization.
\newblock In {\em Europ. Symp. Algorithms (ESA'08)}, volume 5193 of {\em LNCS},
  pages 284--296. Springer, 2008.

\bibitem{cz-upt-13}
M.~Chimani and R.~Zeranski.
\newblock Upward planarity testing via {SAT}.
\newblock In {\em Graph Drawing 2012}, volume 7704 of {\em LNCS}, pages
  248--259. Springer, 2013.

\bibitem{ccdg-bpgms-82}
P.~Z. Chinn, J.~Chv{\'a}talova, A.~K. Dewdney, and N.~E. Gibbs.
\newblock The bandwidth problem for graphs and matrices---a survey.
\newblock {\em J. Graph Theory}, 6(3):223--254, 1982.

\bibitem{deglst-bkvg-07}
A.~M. Dean, W.~Evans, E.~Gethner, J.~D. Laison, M.~A. Safari, and W.~T.
  Trotter.
\newblock Bar $k$-visibility graphs.
\newblock {\em J. Graph Algorithms Appl.}, 11(1):45--59, 2007.

\bibitem{bandwidth-exact}
G.~M. Del~Corso and G.~Manzini.
\newblock Finding exact solutions to the bandwidth minimization problem.
\newblock {\em Computing}, 62(3):189--203, 1999.

\bibitem{DFK+08}
V.~Dujmovic, M.~R. Fellows, M.~Kitching, G.~Liotta, C.~McCartin, N.~Nishimura,
  P.~Ragde, F.~A. Rosamond, S.~Whitesides, and D.~R. Wood.
\newblock On the parameterized complexity of layered graph drawing.
\newblock {\em Algorithmica}, 52(2):267--292, 2008.

\bibitem{minisat}
N.~E{\'e}n and N.~S{\"o}rensson.
\newblock An extensible {SAT}-solver.
\newblock In {\em Proc. Theory and Appl. of Satisfiability Testing (SAT'03)},
  volume 2919 of {\em LNCS}, pages 502--518. Springer, 2004.

\bibitem{ET76}
S.~Even and R.~E. Tarjan.
\newblock Computing an st-numbering.
\newblock {\em Theoret. Comput. Sci.}, 2(3):339--344, 1976.

\bibitem{FLW03}
S.~{Felsner}, G.~{Liotta}, and S.~{Wismath}.
\newblock Straight-line drawings on restricted integer grids in two and three
  dimensions.
\newblock {\em J. Graph Algorithms Appl.}, 7(4):363--398, 2003.

\bibitem{gsm-okpcm-11}
G.~Gange, P.~J. Stuckey, and K.~Marriott.
\newblock Optimal k-level planarization and crossing minimization.
\newblock In {\em Graph Drawing 2010}, volume 6502 of {\em LNCS}, pages
  238--249. Springer, 2011.

\bibitem{gurobi}
{Gurobi Optimization, Inc.}
\newblock Gurobi optimizer reference manual, 2013.

\bibitem{planarityTest}
J.~Hopcroft and R.~Tarjan.
\newblock Efficient planarity testing.
\newblock {\em J. ACM}, 21(4):549--568, Oct. 1974.

\bibitem{jm-2scmpeha-97}
M.~J{\"u}nger and P.~Mutzel.
\newblock 2-layer straightline crossing minimization: Performance of exact and
  heuristic algorithms.
\newblock {\em J. Graph Algorithms Appl.}, 1(1):1--25, 1997.

\bibitem{Kra94}
J.~Kratochv\'{\i}l.
\newblock A special planar satisfiability problem and a consequence of its
  np-completeness.
\newblock {\em Discrete Appl. Math.}, 52(3):233--252, Aug. 1994.

\bibitem{le-tardhpg-03}
X.~Lin and P.~Eades.
\newblock Towards area requirements for drawing hierarchically planar graphs.
\newblock {\em Theoret. Comput. Sci.}, 292(3):679--695, 2003.

\bibitem{mcp-bbambm-08}
R.~Mart{\'\i}, V.~Campos, and E.~Pi{\~n}ana.
\newblock A branch and bound algorithm for the matrix bandwidth minimization.
\newblock {\em Europ. J. of Operational Research}, 186:513--528, 2008.

\bibitem{nw-dlhqm-11}
M.~N{\"o}llenburg and A.~Wolff.
\newblock Drawing and labeling high-quality metro maps by mixed-integer
  programming.
\newblock {\em {IEEE} TVCG}, 17(5):626--641, 2011.

\bibitem{pt-apsto-10}
C.~{Papamanthou} and I.~G. {Tollis}.
\newblock Applications of parameterized st-orientations.
\newblock {\em J. Graph Algorithms Appl.}, 14(2):337--365, 2010.

\bibitem{SZ10}
S.~Sadasivam and H.~Zhang.
\newblock {NP}-completeness of $st$-orientations for plane graphs.
\newblock {\em Theoret. Comput. Sci.}, 411(7-9):995--1003, Feb. 2010.

\bibitem{TT86}
R.~Tamassia and I.~Tollis.
\newblock A unified approach to visibility representations of planar graphs.
\newblock {\em Discrete Comput. Geom.}, 1(1):321--341, 1986.

\bibitem{Wis85}
S.~K. Wismath.
\newblock Characterizing bar line-of-sight graphs.
\newblock In {\em Proc. first Ann. Symp. Comput. Geom.}, SCG '85, pages
  147--152, New York, NY, USA, 1985. ACM.

\end{thebibliography}
  \bibliographystyle{abbrv}
}

\end{document}